\documentclass[conference]{IEEEtran}
\IEEEoverridecommandlockouts
\usepackage{cite}
\usepackage{amsmath,amssymb,amsfonts}
\usepackage{graphicx}
\usepackage{textcomp}
\usepackage{xcolor}
\usepackage{xspace}

\usepackage[capitalise]{cleveref}
\usepackage{url}
\usepackage{amsthm}
\usepackage[numbers,sort&compress]{natbib}
\usepackage{caption}
\usepackage[linesnumbered,ruled,vlined]{algorithm2e}

\newcommand{\system}{Reddio\xspace}
\renewcommand{\paragraph}[1]{\vskip 0.05in \noindent\textbf{#1.}}

\newtheorem{definition}{Definition}
\newtheorem{lemma}{Lemma}
\newtheorem{theorem}{Theorem}

\SetKwProg{Procedure}{Procedure}{ do}{}
\SetKwInOut{KwData}{Global}
\crefalias{algocf}{algorithm}

\def\BibTeX{{\rm B\kern-.05em{\sc i\kern-.025em b}\kern-.08em
    T\kern-.1667em\lower.7ex\hbox{E}\kern-.125emX}}
\begin{document}

\title{Boosting Blockchain Throughput: Parallel EVM Execution with Asynchronous Storage for \system
}

\author{\IEEEauthorblockN{ Xiaodong Qi}
\IEEEauthorblockA{\textit{Nanyang Technological University} }
\and
\IEEEauthorblockN{Xinran Chen}
\IEEEauthorblockA{\textit{Reddio}}
\and
\IEEEauthorblockN{Asiy}
\IEEEauthorblockA{\textit{Reddio} }
\and
\IEEEauthorblockN{Neil Han}
\IEEEauthorblockA{\textit{Reddio} }
}

\maketitle

\pagestyle{plain}

\begin{abstract}
The increasing adoption of blockchain technology has led to a growing demand for higher transaction throughput. Traditional blockchain platforms, such as Ethereum, execute transactions sequentially within each block, limiting scalability. Parallel execution has been proposed to enhance performance, but existing approaches either impose strict dependency annotations, rely on conservative static analysis, or suffer from high contention due to inefficient state management. Moreover, even when transaction execution is parallelized at the upper layer, storage operations remain a bottleneck due to sequential state access and I/O amplification.
In this paper, we propose \system, a batch-based parallel transaction execution framework with asynchronous storage. \system processes transactions in parallel while addressing the storage bottleneck through three key techniques: (i) \emph{direct state reading}, which enables efficient state access without traversing the Merkle Patricia Trie (MPT); (ii) \emph{asynchronous parallel node loading}, which preloads trie nodes concurrently with execution to reduce I/O overhead; and (iii) \emph{pipelined workflow}, which decouples execution, state reading, and storage updates into overlapping phases to maximize hardware utilization.
\end{abstract}

\section{Introduction}

Blockchain technology first gained prominence with Bitcoin~\cite{nakamoto2008bitcoin}, a decentralized cryptocurrency that operates without a central authority. The advent of \emph{smart contracts}~\cite{wood2014ethereum} extended blockchain applications beyond digital currencies to domains such as finance~\cite{financial}, supply chains~\cite{supplychain}, and healthcare~\cite{healthcare}. Smart contracts are self-executing programs deployed on the blockchain, facilitating trustless interactions.  
Ethereum~\cite{wood2014ethereum}, one of the most widely used smart contract platforms, sees transactions involving smart contracts account for nearly 70\% of network traffic. However, Ethereum's execution model presents a major performance bottleneck. Transactions are executed sequentially to maintain state consistency across validators, limiting throughput to around 30 transactions per second (TPS). Additionally, Ethereum employs a computationally intensive \emph{Proof-of-Work} (PoW) consensus mechanism~\cite{nakamoto2008bitcoin}, further constraining scalability. As blockchain adoption grows, overcoming these execution and consensus bottlenecks is critical for improving performance.

With advancements in consensus protocols~\cite{li2020decentralized, yu2020ohie}, blockchain performance bottlenecks are shifting from consensus mechanisms to smart contract execution. Modern protocols such as Conflux~\cite{li2020decentralized} and OHIE~\cite{yu2020ohie} achieve transaction throughputs exceeding 5,000 TPS for simple payments. However, execution remains a major limitation, particularly for smart contracts, where transactions must be processed sequentially to maintain state consistency.  A natural approach to increasing throughput is to pack more transactions into each block without altering the block generation rate. However, this exacerbates execution latency when transactions are processed sequentially. Leveraging multi-core processors for parallel execution offers a potential solution~\cite{amiri2019parblockchain,zhang2018enabling}, but determining whether transactions can be executed in parallel is non-trivial, especially when they involve smart contracts. Ensuring correctness requires \emph{deterministic serializability}, meaning that parallel execution must yield the same result as serial execution in block order. Transactions \emph{conflict} when they access the same state, requiring careful scheduling to avoid inconsistencies. A straightforward approach is to execute non-conflicting transactions in parallel while enforcing serial execution for conflicting ones, balancing performance and correctness.  

Several parallel execution frameworks have been proposed to improve transaction throughput. Some existing approaches~\cite{amiri2019parblockchain,bcos} require explicit read/write set annotations, making them impractical for general-purpose smart contracts where dependencies must be inferred dynamically. Others rely on static analysis tools~\cite{feist2019slither} to extract dependencies, but their coarse-grained analysis often results in overly conservative scheduling, thereby missing significant opportunities for parallel execution.
Another category of parallel execution solutions~\cite{2019Blockchain,sharma2019blurring} employs \emph{optimistic concurrency control} (OCC) to bypass the need for explicit dependency tracking. In OCC, transactions are executed in parallel without dependency enforcement, and after execution, a validation phase detects conflicts. If any transaction violates deterministic serializability, it is aborted and re-executed until no conflicts remain. While OCC improves computational parallelism, excessive abort-and-retry cycles can degrade performance, particularly in workloads with high contention.

Even when parallel execution techniques effectively schedule transactions at the upper layer, their overall benefits remain limited due to a fundamental bottleneck at the storage layer. As highlighted in previous studies~\cite{li2023lvmt}, disk I/O contributes to approximately 70\% of execution overhead due to \emph{I/O amplification}—the phenomenon where reading or writing a small state change necessitates multiple disk accesses. In Ethereum’s conventional execution model, state reads involve traversing the Merkle Patricia Trie (MPT)~\cite{wood2014ethereum} from root to leaf, introducing significant latency. Furthermore, after execution, state updates must be persisted back to storage, often requiring additional sequential operations to update hash values and generate a new root. Since all parallel execution models rely on the same underlying state database, storage access remains a major limiting factor. Even with efficient parallelization at the computation layer, the serialization of state access and updates at the storage layer significantly restricts the actual throughput improvements achieved by these techniques. As blockchain systems scale to millions of accounts and contracts, addressing storage inefficiencies becomes crucial to unlocking the full potential of parallel execution.

In this paper, we propose \system, a batch-based parallel transaction execution framework that effectively addresses both computational and storage bottlenecks in modern blockchain systems. Unlike prior approaches that focus primarily on transaction scheduling while keeping storage operations strictly sequential, \system optimizes execution across both the upper-level transaction processing layer and the lower-level state storage layer, significantly enhancing throughput, scalability, and overall system efficiency.   

\system adopts a \emph{batch-based execution model}, where transactions are grouped into batches and executed in parallel. This design enables efficient scheduling of independent transactions while minimizing state conflicts, ensuring better workload distribution, and reducing synchronization overhead compared to per-transaction parallel execution. However, parallelizing transaction execution alone is insufficient if storage operations remain a bottleneck. To address this, \system introduces an \emph{asynchronous state database} that decouples execution from storage operations, allowing storage updates to be processed efficiently without blocking transaction execution.  To achieve high-performance execution, \system integrates three key techniques into the storage layer:  

\begin{itemize}  
    \item \textbf{Direct state reading:} Instead of traversing the Merkle Patricia Trie (MPT) for each state access, \system allows the EVM to retrieve state values directly from a key-value database, significantly reducing read latency and minimizing I/O amplification.  
    \item \textbf{Asynchronous node retrieval:} To mitigate I/O bottlenecks, \system retrieves necessary trie nodes asynchronously in parallel with transaction execution. By leveraging concurrent database reads, it ensures that storage access does not become a performance bottleneck.  
    \item \textbf{Pipelined state management:} \system introduces a pipelined workflow where transaction execution, state retrieval, and storage updates operate in overlapping phases. This eliminates serialization bottlenecks and maximizes hardware parallelism, further enhancing throughput.  
\end{itemize}  

By integrating batch-based execution with an asynchronous state database, \system effectively maximizes parallel execution efficiency while addressing the storage limitations that traditionally hinder blockchain performance.

\paragraph{Organization}
The remainder of the paper is organized as follows.  
\Cref{sec:background} provides background information and motivates the need for parallel transaction execution in blockchain systems.  
\Cref{sec:overview} presents an overview of \system’s workflow and introduces its key design principles.  
\Cref{sec:design} details the architecture of \system and provides a formal proof of its correctness.  
\Cref{sec:analysis} analyzes the system’s performance, including correctness guarantees and recovery mechanisms.  
Finally, \Cref{sec:conclusion} summarizes our contributions and discusses future research directions.

\section{Background and Motivation}\label{sec:background}

This section provides the necessary background and key concepts required for the rest of the paper.

\subsection{Blockchain and Smart Contracts}

\paragraph{Blockchain}
A \emph{blockchain} is a shared and distributed ledger consisting of a chain of \emph{blocks}, maintained by a decentralized network of nodes. These nodes are categorized as \emph{light nodes}, \emph{full nodes}, or \emph{miners}. Light nodes store only block headers, while full nodes store and validate every block. Miners, a subset of full nodes, participate in block generation by following consensus protocols such as \emph{Proof-of-Work} (PoW)~\cite{nakamoto2008bitcoin}, \emph{Proof-of-Stake} (PoS), and \emph{Proof-of-Authority} (PoA). In this paper, we use ``full node'' and ``node'' interchangeably.

\paragraph{Smart Contracts}
A \emph{smart contract} is a self-executing computer program that enforces user-defined contractual rules on blockchains. Ethereum~\cite{wood2014ethereum} is the most widely used blockchain supporting smart contracts, written in Solidity~\cite{solidity} and compiled into bytecode for execution on the \emph{Ethereum Virtual Machine} (EVM). The EVM is a stack-based machine with a dedicated instruction set. It manages data across three memory areas: persistent storage, contract-local memory (allocated per message call), and the execution stack. \system adopts EVM as its execution environment.

\paragraph{Account Model and Contract States}
\system follows Ethereum’s account model, which includes two types of accounts: \emph{user accounts} and \emph{contract accounts}, both identified by unique 160-bit addresses. A user account holds an Ether balance, lacks executable code, and can initiate transactions. A contract account, in contrast, contains executable code and maintains a storage area, known as the \emph{contract state}. The contract state consists of key-value pairs mapping 256-bit words to 256-bit words. Collectively, the persistent storage of all accounts forms the blockchain state, which is managed by a state database.

Users interact with smart contracts by sending transactions to invoke contract functions. When a contract function is executed, its current state is retrieved from the blockchain, updated, and stored back upon completion. Additionally, \emph{Ether transactions} facilitate direct transfers of Ether between accounts without invoking EVM execution.

\begin{figure}[t]
	\setlength{\abovecaptionskip}{0.2cm}
	\setlength{\belowcaptionskip}{0cm}
	\center{\includegraphics[width=8.5cm]  {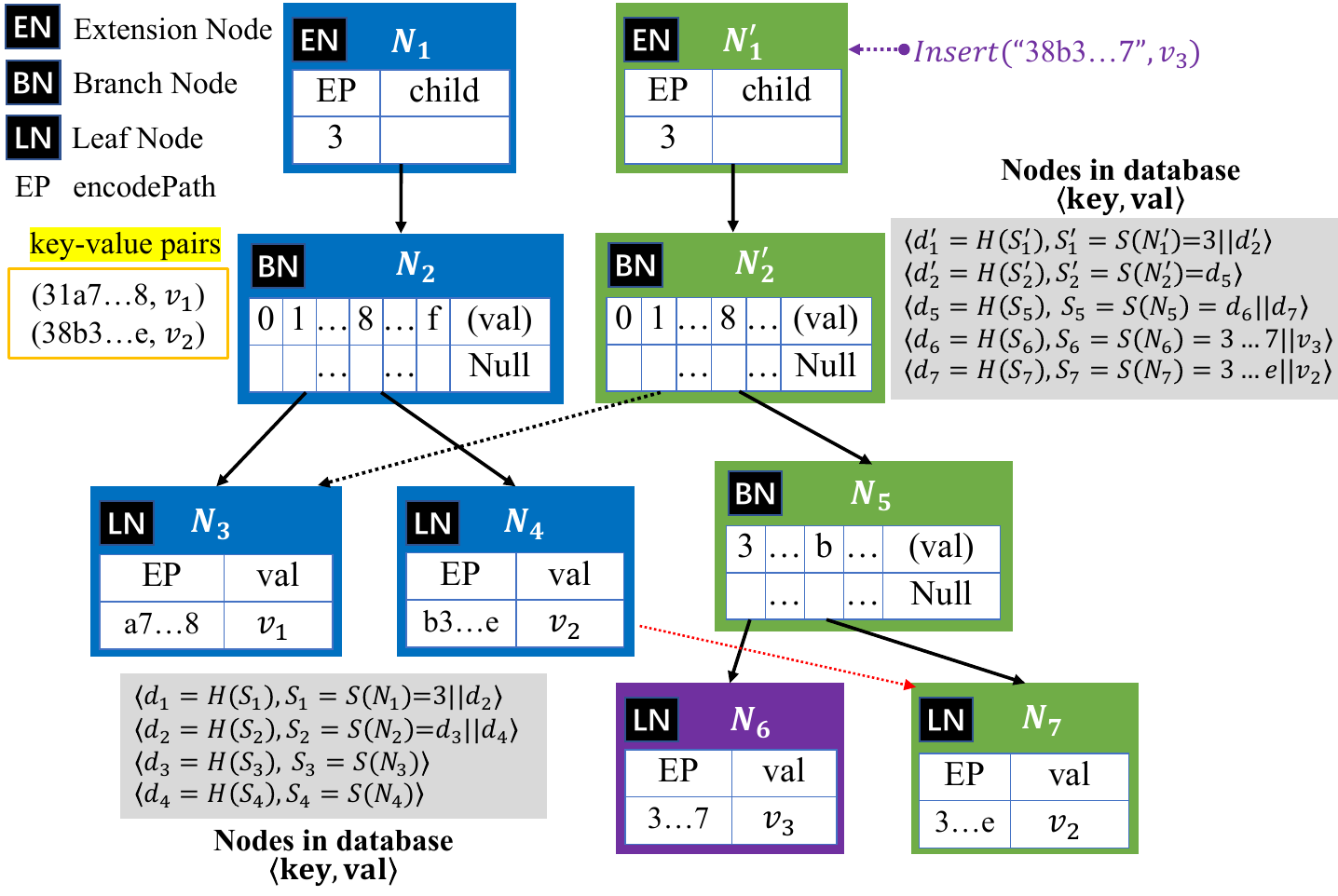}}
	\caption{\label{fig:mpt} An example of Merkle Patrica Trie (MPT).}
\end{figure}

\subsection{State Database} \label{subsec:statedb}

In a typical blockchain system, full nodes synchronize and execute all transactions within blocks, updating the ledger state, which consists of key-value pairs. A cryptographic \emph{hash} of the post-execution ledger state is included in the block header at each block height, enabling light nodes to authenticate state correctness. Additionally, state storage requires data provenance to allow historical state tracing for transaction auditing. In practice, authenticated storage in blockchains is implemented using Merkle trees~\cite{merkle1989certified} and their variants~\cite{kwon2014tendermint, hyperledger}, such as the Merkle Patricia Trie (MPT) used in Ethereum~\cite{wood2014ethereum}.

\paragraph{Merkle Patricia Trie (MPT)}
MPT is a trie with cryptographic authentication, as depicted in \cref{fig:mpt}. A state key is divided into sequential 4-bit characters, called \emph{nibbles}, which guide navigation. MPT consists of three node types: (1) \emph{Branch nodes}, with up to 17 children—16 for different nibble values and one for storing a state value when no further nibbles exist; (2) \emph{Leaf nodes}, containing a value and an \emph{encoded path} for remaining nibbles; (3) \emph{Extension nodes}, which store a shared nibble sequence (encoded path) and point to a single child (e.g., $N_1$ in \cref{fig:mpt}). State retrieval follows a root-to-leaf traversal based on key nibbles.

Each MPT node is stored separately in a key-value database~\cite{o1996log, minglani2017kinetic}, such as LevelDB. A node is uniquely identified by hashing its serialized form $S(N)$ using a cryptographic hash function $H()$~\cite{kecccak256}, and the pair $\langle id_{N}, H(S(N)) \rangle$ is stored in the database. Meanwhile, memory pointers in non-leaf nodes are replaced with corresponding identifiers.

MPT integrates authentication with indexing, using each node's hash as both its identifier and its subtree authentication hash. The root hash ensures integrity across the entire trie. When a state with key $\kappa$ is requested, a full node provides the value and a proof consisting of all traversed nodes from the root to the corresponding leaf. Clients verify state integrity by recomputing hashes up to the root and comparing the result with the block header hash.

\paragraph{State Database}  
Ethereum's state database maintains all account states using a hierarchical Merkle Patricia Trie (MPT) structure for efficient storage, verification, and updates, as illustrated in \cref{fig:account}. At the top level, the \emph{account trie} stores all blockchain accounts, with each account identified by the Keccak-256 hash of its address.  

Ethereum accounts are categorized into two types: user accounts and contract accounts. A user account holds an Ether balance, lacks executable code, and can initiate transactions, while a contract account contains executable code and maintains its own persistent storage in a separate \emph{storage trie}. The state of an account universally consists of \texttt{balance}, \texttt{codeHash}, \texttt{storageRoot}, and \texttt{nonce}. For user accounts, \texttt{codeHash} and \texttt{storageRoot} are empty, as they do not store executable code or contract state.  

The storage trie, specific to contract accounts, maps key-value pairs where keys correspond to the Keccak-256 hash of storage slot indices, and values store associated data. This hierarchical structure ensures that each contract maintains an isolated and verifiable storage space, preventing conflicts while enabling efficient state management.  

When retrieving account-related data, the state database can be accessed at two levels: 
\begin{itemize}
    \item \paragraph{Accessing account information} If only general account data (e.g., balance or nonce) is needed, the query is directly performed on the account trie.
    \item \paragraph{Accessing contract storage} If contract state variables are required, the account trie is first queried to retrieve the \texttt{storageRoot} of the corresponding contract. This \texttt{storageRoot} serves as the root hash of the contract's storage trie, which is then traversed to access specific contract storage values.
\end{itemize}
By structuring account and contract storage in a hierarchical MPT, Ethereum maintains a cryptographically authenticated state while ensuring scalability and security.  

\paragraph{Workflow}
At each block height, the state database operates in three sequential phases: \emph{update}, \emph{hash}, and \emph{store}. During the update phase, the EVM reads and writes state values to the tries in state database for transaction execution. Once all transactions are processed, the state database recalculates the new root hash by recomputing hashes for all \emph{dirty nodes} in account trie or storage tries—those modified during execution. Finally, in the store phase, all dirty nodes are persisted to a underlying key-value database (e.g., LevelDB in current Ethereum). In existing implementation, these three phases are performed sequentially for each block as shown in \cref{fig:syn_workflow}.

\begin{figure}[t]
	\setlength{\abovecaptionskip}{0.2cm}
	\setlength{\belowcaptionskip}{0cm}
	\center{\includegraphics[width=8cm]  {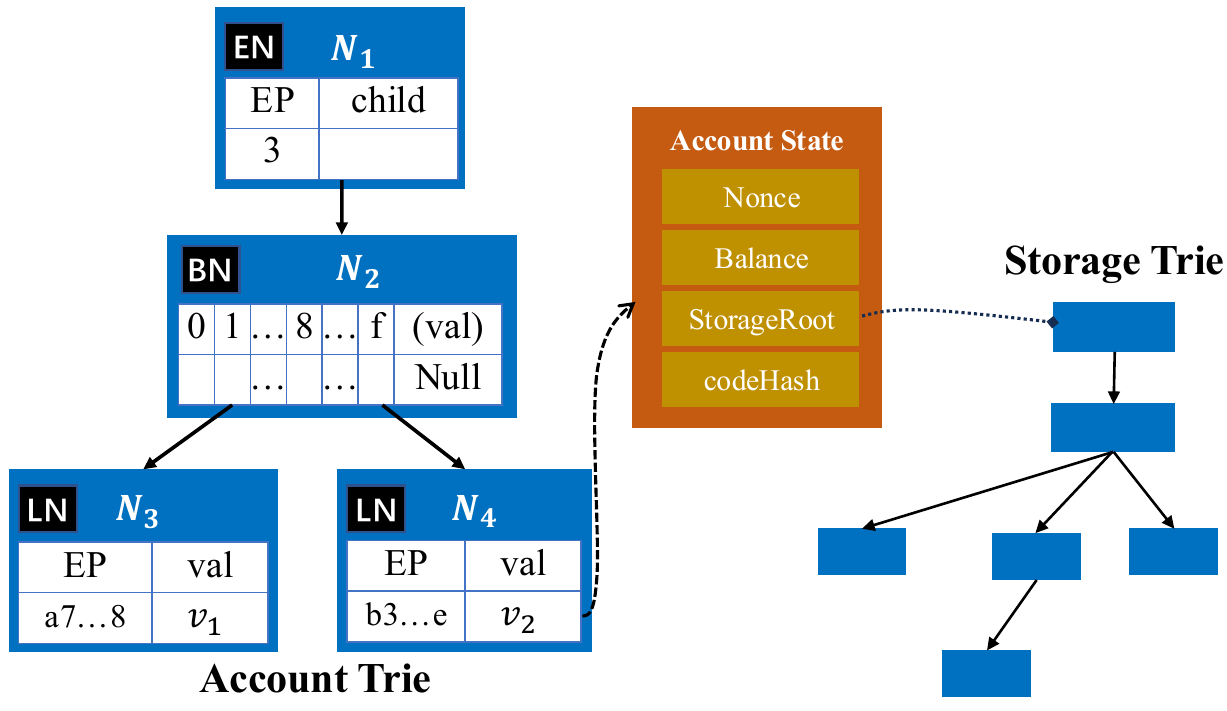}}
	\caption{\label{fig:account} An example of the relationship between the account trie and storage trie.}
\end{figure}

\subsection{Smart Contract Parallel Execution}\label{subsec:backgroud_parallel_execution}

Once a \emph{block} $B_l$, containing a sequence of transactions $\left\langle T_1, \ldots, T_m \right\rangle$ is appended successfully,  nodes execute the transactions in the block, following
the order specified in $B_l$, to update the blockchain state. This serial scheme ensures all nodes reach the consistent blockchain state after execution in current Ethereum, which is critical to the security of it. 
However, this scheme limits the throughput significantly. A direct solution is to leverage the multiple cores available to execute multiple transactions in
parallel, which is well-studied in databases \cite{eswaran1976notions,kung1981optimistic}. These protocols commonly ensure \emph{serializability}, where the effect of concurrent execution is equivalent to a serial execution in \emph{some} order.  The order, however, may vary for different executions, thus nodes, running concurrent execution independently, may enter inconsistent states. 
The parallel executions in blockchain  should additionally meet the \emph{deterministic serializability
criteria}, as defined in~\cref{def:deterministic_serializability}, which promises that all nodes
obtain the same result for every block.

\begin{definition}[Deterministic Serializability]\label{def:deterministic_serializability}
A schedule for a batch of transactions $\langle T_1, \ldots, T_m \rangle$, is \emph{deterministically
serializable} if its effect is equivalent to that of the serialized execution, which conforms to the
transactions' commitment order, $\langle T_1, \ldots, T_m \rangle$.
\end{definition}

Many recent works~\cite{amiri2019parblockchain,2019Blockchain,anjana2019efficient,zhang2018enabling} explore the design space of
parallel transaction execution for smart contracts.
On one hand, some of them assume that the accurate read/write sets of transactions are readily
available, which poses various practical challenges. For example, FISCO BCOS~\cite{bcos} requires users to specify the read/write sets explicitly to support parallelization of transactions. Such a setting is not applicable to smart contracts.
On the other hand, some works~\cite{androulaki2018hyperledger, garamvolgyi2022utilizing} employ the
\emph{Optimistic Concurrency Control} (OCC) strategy to execute transactions in parallel without read/write sets.
With OCC, all transactions read state items from a state snapshot to drive the executions without reading writes of other transactions. As a result,
all transactions can be executed in parallel. After the parallel execution, validators  abort and re-execute the transactions that violate deterministic serializability.

However, according to the reported results~\cite{garamvolgyi2022utilizing}, the speed-up achieved by existing approaches is far from linear on real-world Ethereum workload.  This is mainly due to the lack of inherent parallelism on the real-world workloads---many frequently accessed shared variables force transactions to be executed sequentially, on a few critical paths. These approaches perform  coarse-grained transaction-level concurrency controls without considering the logic of smart contracts, thus they cannot  exploit the potential parallelism by analyzing the state access patterns at the statement level. 
In this paper, we seek to develop an alternative approach that adapts to the existing Ethereum architecture, to achieve much better parallelism by reducing conflicts between transactions.

\begin{figure}[t]
   
	\setlength{\abovecaptionskip}{0.2cm}
	\setlength{\belowcaptionskip}{0cm}
    \center{\includegraphics[width=8cm]  {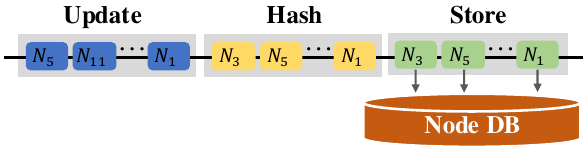}}
	\caption{ \label{fig:syn_workflow} Synchronous workflow in the Merkle Patricia Trie.}
\end{figure}

\subsection{Motivation}\label{subsec:motivation}

Ethereum employs a sequential transaction execution model, where transactions within each block are processed in a strictly ordered manner to ensure deterministic execution and network-wide consensus. However, this approach significantly limits execution efficiency, as even independent transactions without overlapping state must be processed serially. As transaction volume grows, this sequential model becomes a major bottleneck, restricting blockchain throughput and increasing transaction confirmation times.  
With advancements in consensus mechanisms, transaction execution has gradually emerged as the primary performance bottleneck in modern blockchain systems. While consensus algorithms have improved significantly, reducing block propagation delays and finalization times, the efficiency of transaction execution remains a limiting factor. Without optimizing execution, the overall system cannot fully utilize the benefits of faster consensus, ultimately constraining blockchain scalability.

Beyond the limitations of sequential execution, Ethereum's state management further exacerbates performance bottlenecks. The state database, implemented using the Merkle Patricia Trie (MPT), imposes significant computational and I/O overhead due to frequent cryptographic hashing and complex trie operations. Studies indicate that MPT-related operations account for approximately 70\% of Ethereum's transaction execution overhead, as each transaction triggers multiple trie lookups, insertions, and updates, resulting in extensive disk and memory access.  
Furthermore, the synchronous workflow of the state database constrains the potential benefits of other optimizations (\cref{fig:syn_workflow}). Even if transaction execution is parallelized, the sequential nature of state updates remains a performance bottleneck, preventing full realization of parallel execution's advantages.  

To address these challenges, optimizing both transaction execution and state management is essential for improving blockchain's scalability. A parallel execution framework that efficiently handles independent transactions while mitigating state access bottlenecks can significantly enhance blockchain performance, enabling higher throughput and lower latency without compromising consensus guarantees.
\section{Overview}\label{sec:overview}

This section provides a high-level overview of \system's parallel execution and storage optimization.

\begin{figure}[t]
	\setlength{\abovecaptionskip}{0.2cm}
	\setlength{\belowcaptionskip}{0cm}
	\center{\includegraphics[width=8cm]  {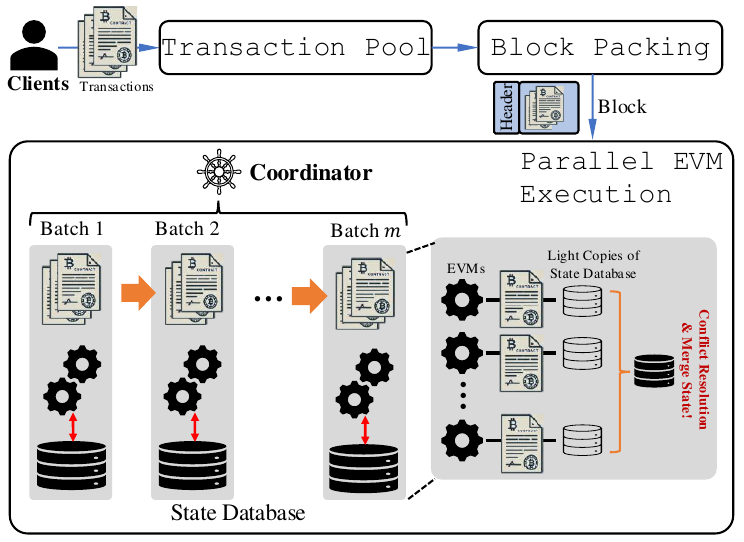}}
	\caption{\label{fig:framework} Framework of parallel EVM execution.}
\end{figure}

\paragraph{Architecture}  
\system enhances transaction throughput by adopting a batch-based parallel execution model and optimizing storage access through asynchronous state management. \cref{fig:framework} illustrates the framework of parallel EVM execution in \system. Each node processes transactions in three stages: (i) transactions are received and stored in the \emph{transaction pool}, (ii) a set of transactions is periodically selected to form a new block, and (iii) the transactions in the block are executed in parallel.  
\system is fully compatible with the EVM ecosystem, integrating the EVM as the execution environment for smart contract transactions. As a result, Ethereum smart contracts can be deployed on \system without modification. This compatibility also ensures that \system retains Ethereum’s account model and state database structure, where account data and contract states are managed as multiple Merkle Patricia Tries (MPTs) as demonstrated in \cref{fig:account}.

To overcome the performance bottlenecks in Ethereum's current transaction execution scheme, \system introduces two core techniques: \emph{parallel EVM execution} and \emph{asynchronous state database}. 
First, to coordinate parallel execution, \system employs a \emph{coordinator} that manages transaction execution batch by batch. In each batch, a set of transactions is selected and assigned to an equal number of worker threads for parallel execution. After each batch completes, the coordinator resolves dependencies among transactions and merges execution outcomes into a new global state database, ensuring deterministic serializability. This process repeats until all transactions are executed and committed.  
Second, \system utilizes an \emph{asynchronous state database} to mitigate I/O bottlenecks, as discussed in \cref{subsec:motivation}. Instead of synchronously updating the state database during execution, \system decouples execution from storage operations, allowing transactions to execute independently while state updates are applied asynchronously.

\paragraph{Parallel execution}  
\system employs an optimistic concurrency strategy to execute transactions in each batch. To facilitate transaction execution on the EVM, every transaction in a batch operates on a lightweight copy of the global state database. During execution, transactions proceed independently under the assumption that no conflicts exist.  
To maximize parallelism, each batch should ideally contain transactions with minimal or no dependencies. However, conflicts are inevitable. Therefore, after executing each batch, the coordinator detects conflicts and aborts any conflicting transactions. Let a block at height \( l \) contain a set of transactions \( B_l = \langle T_1, \ldots, T_{\alpha} \rangle \). The notion of a \emph{conflict} between two transactions is formally defined as follows.  

\begin{definition}[Transaction Conflict]\label{def:conflict}  
Given two transactions \( T_i \) and \( T_j \) (\( i < j \)), a conflict occurs if \( T_i \) writes to a state item that \( T_j \) subsequently reads.  
\end{definition}  
\noindent According to \cref{def:conflict}, if a transaction \( T_j \) conflicts with a preceding transaction \( T_i \) (\( i < j \)), it must be aborted. This is because \( T_j \) reads the state value from the snapshot rather than the updated value written by \( T_i \), violating the serializability. The coordinator then commits the remaining transactions and merges their state database copies into a new global state database, providing an updated snapshot for the next execution round. Aborted transactions are deferred to subsequent batches for re-execution. This iterative process continues until all transactions have been successfully executed.

\paragraph{Asynchronous state database}  
While parallel execution enhances transaction throughput, its benefits diminish as the number of execution threads increases. Beyond a certain threshold, I/O overhead becomes the primary bottleneck. As discussed in \cref{subsec:motivation}, read and write operations to the state database contribute over 70\% of total execution overhead due to the significant I/O amplification caused by MPTs. Consequently, pure parallel execution alone, without optimizing the underlying storage system, cannot fully exploit the potential of parallelism.  

\begin{figure}[t]
    
	\setlength{\abovecaptionskip}{0.2cm}
	\setlength{\belowcaptionskip}{0cm}
	\center{\includegraphics[width=8.5cm]  {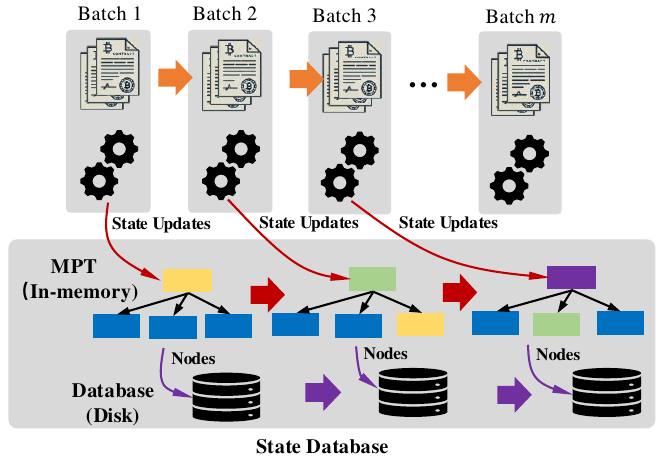}}
	\caption{\label{fig:pipeline_workflow} Asynchronous pipeline of state database in \system.}
\end{figure}

To address this limitation, \system enhances the state database by introducing \emph{direct state reading}, \emph{asynchronous node retrieval}, and a \emph{pipelined workflow}.  

First, in the current state database, when execution requests a state value, it must locate the root hash of the corresponding trie and traverse the trie to retrieve the value. This process requires loading all nodes along the path from the root to the leaf, leading to multiple rounds of I/O amplification. However, since these structures primarily serve to authenticate state integrity rather than provide direct execution semantics, the EVM does not need to interact with intermediate trie nodes. Therefore, \system enables the EVM to access requested state values directly, bypassing intermediate node retrievals.  
Second, when a state is updated, all nodes along its search path must be loaded for hash recomputation. To mitigate I/O blocking, \system records the new state value in memory and loads these nodes asynchronously.  

Third, despite these optimizations, the asynchronous workflow of the state database still imposes performance constraints, as shown in \cref{fig:syn_workflow}. In particular, the hash and store phases must be executed synchronously after the update phase, which is directly tied to transaction execution. This sequential dependency prevents throughput from scaling linearly with the number of execution threads, creating a new bottleneck.  
To address this issue, \system adopts a pipelined workflow, as illustrated in \cref{fig:pipeline_workflow}. After each batch execution, \system selectively rehashes and persists certain nodes in the account trie and storage tries to the underlying key-value database (e.g., LevelDB), forming a pipeline between execution and state persistence. However, premature hashing of frequently modified nodes may lead to redundant computation and unnecessary storage overhead. To mitigate this, \system dynamically determines an \emph{optimal point} for early hashing, ensuring that recalculations remain infrequent while maintaining storage efficiency. This pipeline design enables better resource utilization and sustains high transaction throughput.  

\section{Design}\label{sec:design}

\subsection{Batch-based Optimistic Parallel Execution} \label{subsec:parallel_exec}

\begin{algorithm}[t]
	\small
	\setlength{\belowdisplayskip}{1pt}
	\caption{Transaction batch fetching \label{algorithm:batch_fetch}}
    \KwIn{Transaction set $\mathcal{T}$ in block $B_l$, thread number $\eta$}
    \KwOut{A batch of  transactions $\mathcal{T}_{batch}$}
    $\mathcal{T}_{batch}\gets\emptyset$\\
    \For{$i\gets0;i<\eta; i\gets i+1$}{
        $T_{p_i}\gets \texttt{NextTx}(\mathcal{T}, \mathcal{T}_{batch})$\\
        \If{\rm $T_{p_i}\neq \texttt{null}$}{
            $\mathcal{T}_{batch}\gets \mathcal{T}_{batch}\cup\{T_{p_i}\}$\\
            $\mathcal{T}\gets \mathcal{T}\setminus\{T_{p_i}\}$\\
        }\Else{
            \textbf{break}
        }  
    }
    \If{\rm $len(\mathcal{T}_{batch})<\eta$}{
        $\mathcal{T}_{batch}\gets \texttt{fill\_Txs}(\mathcal{T}_{batch})$
    }
    \Return{$\mathcal{T}_r$}\\
    
    \vspace{5pt}
    \Procedure{\rm $\texttt{NextTx}(\mathcal{T},\mathcal{T}_{batch})$}{
        $T_{next}\gets \texttt{null}$\\
        \For{$i\gets 0; i<len(\mathcal{T});i\gets i+1$ }{
            $T_{p_i}\gets \mathcal{T}[i]$\\
            $r\gets true$\\
            \ForEach{$T_{p_j}\in \mathcal{T}_{batch}$}{
                \If{\rm $\texttt{explicit\_conflict}(T_{p_i}, T_{p_j})$ is true}{
                    $r\gets false$
                }
            }
            \If{\rm $r$ is \textbf{true}}{
                $T_{next}\gets T_{p_i}$\\
                \textbf{break}
            }
        }

        \Return{$T_{next}$}
    }
\end{algorithm}

As previously discussed, \system adopts a batch-based optimistic parallel execution scheme, where each transaction executes independently on a copy of the global state database without any cross-thread communication. At the beginning of each batch, a set of transactions is selected for execution in the next round.

\begin{algorithm}[t]
    \small
    \setlength{\belowdisplayskip}{1pt}
    \caption{Parallel execution of transaction batch \label{algorithm:batch_parallel_exec}}
    \KwIn{ Transaction batch $\mathcal{T}_{batch}$, global state database $S$, set of remained transactions $\mathcal{T}$, set of reads of executed transactions $\mathcal{R}$, set of writes of executed transactions $\mathcal{W}$, index of next transaction to be committed $I_{next}$ }
    \KwOut{Global state $S$ after executing transactions in $\mathcal{T}_{batch}$, set of state updates written by successfully committed transactions $W$, set of reads of executed transactions $\mathcal{R}$, set of writes of executed transactions $\mathcal{W}$, index of next transaction to be committed $I_{next}$  }
    \For{\rm $i\gets 0; i<len(\mathcal{T}_{batch});i\gets i+1$}{ 
        \tcp{Execute each transaction on a separate thread}
        $T_{p_i}\gets \mathcal{T}_{batch}[i]$\\
        $S_{p_i}\gets \texttt{light\_Copy}(S_{p_i})$\\
        $R_{p_i}, W_{p_i}\gets \texttt{parallel\_Execute}(T_{p_i}, S_{p_i})$\\

        $\mathcal{R}\gets \texttt{append}(\mathcal{R}, R_{p_i})$, $\mathcal{W}\gets \texttt{append}(\mathcal{W},W_{p_i})$\\
    }
    \tcp{Wait for all threads to terminate}
    $\texttt{wait()}$\\
    $\texttt{sort}(\mathcal{R})$, $\texttt{sort}(\mathcal{W})$\\
    
    \tcp{Merge states produced by all threads}
    
    $W\gets \emptyset$\\
    \For{$i\gets 0; i<len(\mathcal{R}); i\gets i+1$}{
        \If{\rm $R_{p_i}$ overlaps with $W_{p_0}, \ldots, W_{p_{i-1}}$}{
            $\texttt{abort}(\mathcal{T}_{batch}[i])$\\
            $\mathcal{T}\gets \mathcal{T}\cup \{\mathcal{T}_{batch}[i]\}$\\
        } \ElseIf{$p_i=I_{next}$}{
            $S_{merge}\gets \texttt{merge\_State}(S_{merge}, W_{p_i})$\\
            $W\gets W\cup W_{p_i}$\\
            $\mathcal{R}\gets \mathcal{R}\setminus \{R_{p_i}\}$, $\mathcal{W}\gets \mathcal{W}\setminus \{W_{p_i}\}$\\
            $I_{next}\gets I_{next}+1$
        }
    }
    $S\gets S_{merge}$\\
    \Return{$S,W, \mathcal{R}, \mathcal{W}, I_{next}$}
\end{algorithm}

\paragraph{Transaction batch fetching}  
The coordinator is responsible for constructing each batch, as presented in \cref{algorithm:batch_fetch}. To form a batch, the coordinator first identifies the accounts accessed by each transaction. Let \( \mathcal{T} \) denote the set of remaining transactions, which initially includes all transactions in the current block, and let \( \eta \) represent the number of available execution threads.  
A straightforward approach is to select the \( \eta \) transactions in \( \mathcal{T} \) with the smallest indexes, ensuring deterministic serializability. However, this method may introduce excessive conflicts within a batch, limiting parallelism. Instead, the coordinator in \system carefully selects transactions with fewer conflicts to maximize parallel execution efficiency.

The process starts by inserting the first transaction from \( \mathcal{T} \) into \( \mathcal{T}_{batch} \). Then, for each subsequent transaction \( T_{p_i} \), \system checks whether it explicitly conflicts with any transactions already in \( \mathcal{T}_{batch} \) (Lines 1-8). To maximize concurrency, transactions with explicit conflicts should be avoided during the fetching phase. Recall that each transaction \( T_{p_i} \) has a sender account \( T_{p_i}.s \) and a receiver account \( T_{p_i}.r \), both of which must be accessed during execution. To reduce conflicts, the coordinator ensures that transactions added to \( \mathcal{T}_{batch} \) do not explicitly overlap in sender or receiver accounts (Lines 15-22).  
Specifically, a transaction \( T_{p_i} \) is considered to have an explicit conflict with an existing transaction \( T_{p_j} \) in \( \mathcal{T}_{batch} \) if any of the following conditions hold:  
\begin{enumerate}  
    \item \( T_{p_i}.s = T_{p_j}.s \) or \( T_{p_i}.s = T_{p_j}.r \), meaning \( T_{p_i} \)'s sender overlaps with \( T_{p_j} \)'s sender or receiver.  
    \item \( T_{p_i}.r = T_{p_j}.s \) or \( T_{p_i}.r = T_{p_j}.r \), meaning \( T_{p_i} \)'s receiver overlaps with \( T_{p_j} \)'s sender or receiver.  
\end{enumerate}  
If such a conflict is detected, \( T_{p_i} \) is skipped in this batch selection to avoid foreseeable conflicts (Lines 18-19). If no significant conflicts are found, the transaction is added to the batch (Lines 20-22). This selection process continues until the batch reaches size \( \eta \).  
If the final batch size remains smaller than \( \eta \), the coordinator forcibly fills the batch with the remaining transactions having the smallest indexes from \( \mathcal{T} \), regardless of potential conflicts. This ensures that all available execution threads are utilized, even at the cost of some potential rollbacks (Lines 9-10).

\paragraph{Parallel execution} Once obtains the batch $\mathcal{T}_{batch}$ of transactions, the coordinator schedules the parallel execution of them, as illustrated in \cref{algorithm:batch_parallel_exec}. It takes the transaction batch $\mathcal{T}_{batch}$, a global state database $S$, and the set of remained transactions $\mathcal{T}$, the read/write sets $\mathcal{R}$/$\mathcal{W}$ of transactions that have been executed, and the index $I_{next}$ of next transaction to be committed as input and returns a new global state database after executing transactions, and a set of state updates $W$ written by successfully committed transactions. 

The coordinator begins by assigning each transaction \( T \in \mathcal{T}_{batch} \) to a worker thread, initiating an EVM instance for execution (Lines 1-5). Each worker thread executes its assigned transaction in parallel on a separate copy of the state database \( S_{p_i} \) (Line 3). This approach ensures that transactions run independently on different threads without direct interactions.  
During execution, each thread records the read set \( R_{p_i} \) and write set \( W_{p_i} \) accessed by its transaction (Line 4), which are appended to $\mathcal{R}$ and $\mathcal{W}$ respectively. Although \( \mathcal{T}_{batch} \) is carefully selected to avoid explicit conflicts, implicit conflicts may still arise. For example, multiple transactions may invoke the same smart contract internally, leading to unintended dependencies not immediately evident from their senders and receivers. Such conflicts require additional coordination to maintain correctness. To cope with this issue, the read/write sets in $\mathcal{R}$/$\mathcal{W}$ are sorted by the transaction index in a ascending order (Line 7). 
If the read set \( R_{p_i} \) overlaps with any write set \( W_{p_0}, \ldots, W_{p_{i-1}} \), then \( \mathcal{T}_{batch}[i] \) has read stale data, missing updates written by preceding transactions. In this case, the coordinator aborts the transaction and returns it to \( \mathcal{T} \) for future re-execution (Lines 10-12). Otherwise, if the transaction index $p_i$ equals $I_{next}$, the transaction is committed, and the coordinator merges its state updates into the new state database \( S_{merge} \) (Lines 13-17). Finally, the updated state database and write set are returned.  

\begin{algorithm}[t]
    \small
    \setlength{\belowdisplayskip}{1pt}
    \caption{State database operation \label{algorithm:statedb_op}}
    \KwData{State database $S$, memory state cache $\mathcal{C}_{state}$, task queue $Q_{ret}$}
    \Procedure{\rm $\texttt{Get}(A, \kappa)$}{
        $v\gets \mathcal{C}_{state}.\texttt{get}(A||\kappa)$\\
        \If{\rm $v= \texttt{null}$}{
            $D_{direct}\gets$ direct database in $S$ for direct state reading\\
            $v\gets \texttt{direct\_Get}(D_{direct}, A||\kappa)$\\
            $\mathcal{C}_{state}.\texttt{set}(A||\kappa, v)$
        }
        \Return{$v$}
    }
    \Procedure{\rm $\texttt{Set}(A, \kappa, v)$}{
        $\mathcal{C}_{state}.\texttt{set}(A||\kappa, v)$\\
        $o\gets \langle A, \kappa, v\rangle$\\
        $Q_{ret}.\texttt{Push}(o)$
    }
\end{algorithm}

\subsection{Asynchronous State Database} 
While parallel execution improves computation, storage access remains a major bottleneck due to its inherently sequential nature. In \system, although each worker thread operates on a lightweight copy of the state database, they share the same underlying key-value store, leading to sequential read dependencies across threads. As the state database scales to millions of accounts, execution is no longer the primary limiting factor—storage latency dominates, restricting overall speedup.  
To mitigate I/O bottlenecks, \system employs an asynchronous workflow that decouples execution from storage. However, this approach introduces new challenges, including stale data access and unpredictable write bursts that can overwhelm the storage backend. To address these issues, \system strategically schedules state commits, balancing execution efficiency and storage consistency.  To this end, \system introduces three key optimizations: \emph{direct state reading}, \emph{asynchronous node retrieval}, and \emph{pipelined workflow}, as detailed later in this subsection. These techniques collectively mitigate sequential state access bottlenecks, improving execution efficiency and overall system throughput.

\paragraph{Direct state reading}  
During transaction execution, the EVM does not require state proofs, yet conventional state databases still traverse all nodes from the root to the leaf for each read operation, whether in the account trie or storage trie, as described in \cref{subsec:statedb}. This process incurs significant latency, as each access involves multiple disk I/O operations, slowing EVM execution.  

To reduce this overhead, \system enables the EVM to retrieve state values directly from a separate key-value database \( D_{direct} \) with a single read operation. For an account with address \( A \) and state value \( \mathcal{V} \), the state database maintains an additional record \( \langle A, \mathcal{V} \rangle \) in \( D_{direct} \). Similarly, for each state entry \( \kappa \) of a smart contract account \( A \) with value \( v \), the state database stores \( \langle A || \kappa, v \rangle \), where ``$||$'' denotes concatenation. As illustrated in \cref{fig:direct_read_asyn_node_ret}, when execution needs to read the value of account ``38...e'', it directly retrieves the value from \( D_{direct} \), bypassing the account trie.

\begin{figure}[t]
	\setlength{\abovecaptionskip}{0.2cm}
	\setlength{\belowcaptionskip}{0cm}
	\center{\includegraphics[width=8.5cm]  {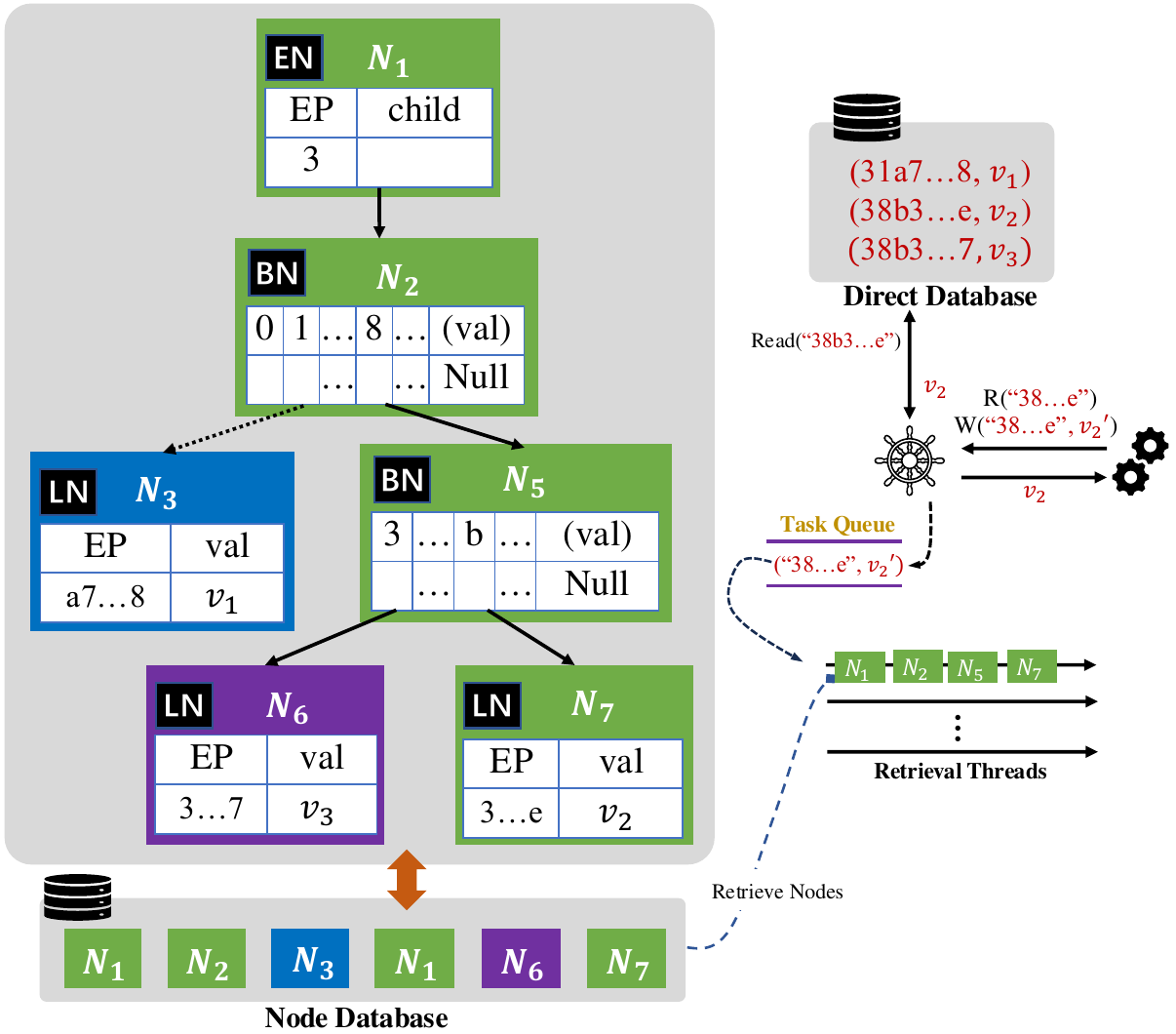}}
	\caption{\label{fig:direct_read_asyn_node_ret} Example of direct state reading and asynchronous node retrieval in \system.}
\end{figure}

The \texttt{Get()} function in \cref{algorithm:statedb_op} details this process. Three global data structures are shared across algorithms in this paper: the state database \( S \), the memory state cache \( \mathcal{C}_{state} \), and a queue \( Q_{ret} \). The queue \( Q_{ret} \) is used for asynchronous node retrieval, which will be detailed next.  
During execution, transaction updates are first written to the state cache \( \mathcal{C}_{state} \) before being persisted at an appropriate time. The function \texttt{Get()} takes an account address \( A \) and a state key \( \kappa \) as input. If \( \kappa \) is empty, the request targets the storage of an account; otherwise, it retrieves the value of state \( \kappa \) within contract \( A \). If the requested state is not found in \( \mathcal{C}_{state} \), \texttt{Get()} retrieves the value directly from \( D_{direct} \), bypassing both the account and storage tries (Lines 3-6).

While direct state reading improves retrieval efficiency by bypassing MPT traversals, it introduces the challenge of maintaining consistency between the fast-access state database and the MPT structure, especially during state updates or recovery from unexpected crashes. The solution to this issue is detailed in \cref{subsec:recovery}.

\paragraph{Asynchronous node retrieval}  
While direct state reading optimizes read operations, write operations still require modifying the state, triggering node modifications and rehash calculations along the search path. Retrieving these nodes only after a write operation introduces significant I/O overhead, blocking execution.  
To address this issue, \system adopts an asynchronous approach where state updates are first buffered in \( \mathcal{C}_{state} \), while dedicated threads load the required nodes asynchronously in parallel with execution. To further accelerate this process, \system leverages the concurrent reading capabilities of key-value databases, such as LevelDB. Specifically, the state database initializes \( \zeta \) dedicated load threads to retrieve nodes concurrently.  

As illustrated in \cref{fig:direct_read_asyn_node_ret}, when execution updates a state value \( v_2' \) for account ``38...e'', the new value is first cached in memory, and a retrieval task is pushed into the queue \( Q_{ret} \). The retrieval threads then continuously dequeue tasks from \( Q_{ret} \) and concurrently fetch the corresponding MPT nodes from the underlying node database. The \texttt{Set()} function in \cref{algorithm:statedb_op} formalizes this write operation. When a write occurs, the updated value is first cached in \( \mathcal{C}_{state} \) (Line 9). Then, the write operation, represented as \( o = \langle A || \kappa, v \rangle \), is inserted into the task queue \( Q_{ret} \) (Lines 10-11). If \( \kappa \) is empty, \( o \) updates the account storage of \( A \); otherwise, it modifies the contract state at key \( \kappa \).

The main logic of asynchronous node retrieval is presented in \cref{algorithm:asyn_node_retrieval}. A set of \( \zeta_{r} \) retrieval threads continuously fetch tasks from \( Q_{ret} \) and retrieve the corresponding nodes along the search paths (Lines 1-9). Each thread begins by dequeuing a task \( o = \langle A, \kappa, v \rangle \) from \( Q_{ret} \), then accesses the account trie in \( S \) and loads the necessary nodes to identify account \( A \) using the \texttt{load\_Nodes()} function. This function loads nodes from an MPT (either the account trie or a storage trie) stored in the underlying key-value database.  
If \( \kappa \) is null, \( o \) updates only the state storage of account \( A \). Otherwise, since the write operation modifies contract state, the thread must also load nodes from \( A \)'s storage trie. The \texttt{load\_Nodes()} function retrieves nodes along the search path of the specified trie and caches them in the node cache \( \mathcal{C}_{node} \) (Lines 10-21). If a requested node is not found in the cache, it is fetched from the node database \( D_{node} \) (Lines 14-16).  
Since the underlying key-value database supports concurrent reads, multiple retrieval threads can significantly improve performance by parallelizing state access.

Since the sender and receiver accounts of a transaction are known in advance, their corresponding entries in the account trie can be preloaded before execution, reducing lookup latency. In contrast, for smart contract invocations, the specific storage locations accessed depend on runtime execution logic and cannot be determined statically. As a result, contract storage nodes must be loaded dynamically during execution.

\begin{algorithm}[t]
    \small
    \setlength{\belowdisplayskip}{1pt}
    \caption{Asynchronous node retrieval \label{algorithm:asyn_node_retrieval}}
    \KwData{State database $S$, memory node cache $\mathcal{C}_{node}$, retrieval task queue $Q_{ret}$}

    \tcp{Run continuously on a separate thread}
    \While{\rm true}{ \label{alg:asyn_node_ret_start}
        $o=\langle A, \kappa, v\rangle \gets Q_{ret}.\texttt{Pop()}$\\

    $D_{node}\gets$ node database in $S$ for nodes\\
    $Trie_{acc}\gets$ account trie in $S$\\
    $\mathcal{V}\gets \texttt{load\_Nodes}(Trie_{acc}, A, D_{node}, \mathcal{C}_{node})$\\
    \If{\rm $\kappa=\texttt{null}$}{
        \Return{}
    }
    $Trie_{storage}\gets storage\_Trie(\mathcal{V}, A)$
    \\
    $\texttt{load\_Nodes}(Trie_{storage}, \kappa, D_{node}, \mathcal{C}_{node})$ \label{alg:asyn_node_ret_end}\\
    }
    \vspace{5pt}
    \Procedure{\rm $\texttt{load\_Nodes}(Trie, key, D_{node}, \mathcal{C}_{node})$}{
        $next\gets$ the root hash of MPT $Trie$\\
        $node\gets \texttt{null}$\\
        \While{\rm $next\ne \texttt{null}$}{
            \If{\rm node $next$ is not in $\mathcal{C}_{node}$}{
                $node\gets \texttt{node\_Retrieve}(D_{node}, next)$\\
                $\texttt{set\_Node}(\mathcal{C}_{node}, next, node)$\\
            }\Else{
                $node\gets \texttt{get\_Node}(\mathcal{C}_{node}, next)$\\
            }
            \If{\rm $node$ is the leaf}{
                \Return{\rm the value in $node$ }
            }
            $next\gets \texttt{next\_Node}(node, A)$\\
        }
    }
\end{algorithm}

\begin{algorithm}[t]
    \small
    \setlength{\belowdisplayskip}{1pt}
    \caption{Pipelined workflow in state database \label{algorithm:pipeline_workflow}}
    \KwData{task queue for hash $Q_{hash}$, task queue for store $Q_{store}$, global state database $S$}

    \Procedure{\rm \texttt{AsyCommitAccount$()$}}{
        $\mathcal{N}\gets$ nodes in the account trie that have reached commit points without dirty children\\
        \For{$N$ {\rm \textbf{in}} $\mathcal{N}$}{
            $Q_{hash}.\texttt{Push}(N)$
        }
    }

   \Procedure{\rm \texttt{HashThread$()$} }{ 
        \For{\rm true} {
            $N\gets Q_{hash}.\texttt{Pop}()$  \\
            \For{\rm $\texttt{can\_Hash}(N)$} {  
                    \If{\rm $d\gets 
                    \texttt{Hash}(N)$ is not $null$}{ 
                        $\texttt{set\_Hash}(N,d)$\\
                        $Q_{store}.\texttt{Push}(N)$ \\
                    } \Else { 
                        \textbf{break} 
                    }
                $N\gets$ parent node of $N$ \\  
            }
        }
   }

   \Procedure{\rm \texttt{StoreThread$()$}}{ 
        $D_{node}\gets$ node database in $S$ for nodes\\
        \For{\rm true}{
            $N\gets Q_{store}.\texttt{Pop}()$ \\
            \If{\rm $ctx\gets \texttt{Ser}(N)$ is not $null$}{
                $\texttt{store\_Node}(D_{node}, ctx, N)$\\
                }
        }
   }
\end{algorithm}

\paragraph{Pipelined workflow}  
As discussed in \cref{subsec:motivation}, the synchronous workflow of the state database significantly limits the benefits of parallel execution. To address this issue, \system introduces a pipelined workflow that overlaps state updates with transaction execution to improve efficiency.  
The state update process consists of three sequential phases: update, hash, and store. In the conventional workflow, these phases are executed serially, meaning that each phase must fully complete before the next phase begins. When a state is modified, all nodes along its search path must be updated, rehashed, and then persisted in the key-value database. Since hashing depends on updated node values and storage depends on finalized hashes, these dependencies force a strict execution order, preventing concurrent processing and introducing performance bottlenecks.  

\system optimizes the workflow by overlapping hashing and storage with transaction execution. Instead of waiting until execution completes, modified nodes are rehashed and persisted asynchronously, reducing the delay between state updates and final storage. However, if a node is modified multiple times by later transactions, premature rehashing may lead to redundant computation.  \system employs different asynchronous rehashing strategies based on trie type, ensuring efficient state updates while minimizing unnecessary rehashing.   

\begin{itemize}  
    \item \paragraph{Storage trie} If a smart contract account is unlikely to be accessed again within the block, meaning it will not be explicitly accessed by the remaining transactions, \system rehashes its storage trie and persists the modified nodes.  
    \item \paragraph{Account trie} Once a storage trie is rehashed, its corresponding leaf node in the account trie is updated, prompting \system to rehash the account trie at the node level. For each node in the account trie, \system estimates a commit point---after which the node is unlikely to be modified---allowing for early rehashing.  
\end{itemize}   
\noindent  \system differentiates between the two by adopting distinct asynchronous processing strategies tailored to their access patterns.  
Storage tries are processed asynchronously at the account level because they exhibit lower contention and their access patterns can be explicitly predicted based on transaction execution. If a contract account is not expected to be accessed again within the block, its storage trie can be speculatively rehashed and persisted.  
In contrast, the account trie is frequently accessed and modified by multiple transactions. To reduce redundant computations from concurrent modifications, \system employs a finer-grained asynchronous pipeline for the account trie, determining a commit point for each node individually. By delaying rehashing until a node is unlikely to be modified again, this approach improves efficiency while maintaining consistency.

After each batch execution, the coordinator checks whether a contract account \( A \) will be explicitly accessed by any remaining transactions. If not, \system speculatively rehashes and persists the Merkle Patricia Trie (MPT) of \( A \)'s storage into the node database. Although an internal transaction call could still access \( A \) implicitly, potentially causing unnecessary rehashing, such cases are rare in practice and are considered an acceptable trade-off.  
The following subsection details this process, while this section focuses on the pipelined workflow for the account trie.   

\begin{figure}[t]
	\setlength{\abovecaptionskip}{0.2cm}
	\setlength{\belowcaptionskip}{0cm}
    \center{\includegraphics[width=6.5cm]  {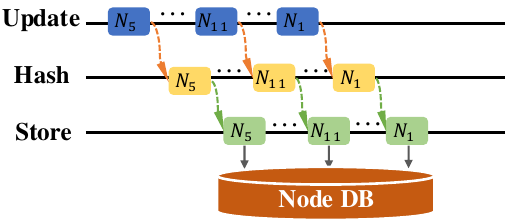}}
	\caption{\label{fig:pipeline_account_trie}Pipelined workflow for asynchronous state database.}
\end{figure}

\system employs two additional threads to handle the hash and store phases in parallel, while the update phase is triggered by the commit thread, as described in \cref{algorithm:framework}. Ideally, when a node in the account trie completes one phase, it is immediately passed to the next, enabling concurrent processing across phases. However, nodes—especially those closer to the root—may undergo multiple modifications, preventing immediate commitment to the hash phase. To mitigate this, \system estimates a \textit{commit point} for each node, representing a threshold beyond which further modifications are unlikely (e.g., less than 10\%). Once this commit point is reached, \system speculatively commits the node for hashing.  

The pipelined workflow is outlined in \cref{algorithm:pipeline_workflow}. At the end of each batch execution, the coordinator invokes $\texttt{AsynCommitAccount}()$, which pushes all nodes in the account trie that have reached their commit points into a queue \( Q_{hash} \) (Lines 1-4). The $\texttt{HashThread}()$ continuously dequeues nodes from \( Q_h \), recalculates their hashes, and, once a node \( N \) is hashed, attempts to hash its parent if all its child nodes have already been processed (Lines 5-14). The hashed nodes are then pushed into queue \( Q_{store} \) for the subsequent storage phase. Finally, the $\texttt{StoreThread}()$ serializes the nodes dequeued from \( Q_{store} \) and writes them into the node database \( D_{node} \) (Lines 15-19).  Moreover, the determination of each node's commit point is discussed in \cref{subsec:commit_point}.

In the asynchronous workflow of the account trie, nodes are accessed concurrently by the update, hash, and store phases as shown in \cref{fig:pipeline_account_trie}. If a node is being hashed while another modification occurs, the hash phase may process an outdated or intermediate state. To prevent this, \system assigns a lock to each node, ensuring exclusive access during processing. Additionally, before a thread begins hashing or persisting a node, it checks whether the node has been modified since its last update. If further modifications have occurred, the hash or store phase is aborted to avoid processing stale data. For simplicity, this mechanism is not explicitly presented in the algorithm.

\subsection{Framework}  

In this subsection, we present the overall framework of \system's parallel EVM execution with asynchronous storage, as illustrated in \cref{algorithm:framework}. The algorithm is managed by the main thread (coordinator). First, the coordinator initializes the required threads, 
$\zeta_{r}$ retrieve threads to retrieve nodes, a hash thread, and a store thread (Lines 1-4).  The thread $\texttt{AsyCommitStorage}()$ is responsible for committing the storage tries of smart contracts. Initially, the index of the next transaction to be committed is 0, and the set of read or write sets generated by transactions that have been executed is also set to empty (Line 5). 
The main loop iterates batch by batch until all transactions have been processed (Lines 7-17).  
In each iteration, the coordinator fetches a batch of transactions, assigns \( \eta \) worker threads for concurrent execution, and collects the resulting updates \( \overline{W} \), which are generated by committed transactions (Lines 8-10). For each account \( A \) recorded in \( \overline{W} \), if \( A \) will not be accessed explicitly, the coordinator flushes all updates to \( A \)'s storage trie by pushing it into \( Q_{commit} \) for asynchronous commitment. The \( \zeta_c \) commit threads continuously acquire accounts from \( Q_{commit} \) and apply updates to storage tries asynchronously (Lines 23-30). Afterward, the coordinator invokes \texttt{AsynCommitAccount}() to trigger the pipelined workflow of the account trie as defined in \cref{algorithm:pipeline_workflow}. Finally, it waits for all threads to complete the commit workflow before returning the updated state database.

\begin{algorithm}[t]
    \small
    \setlength{\belowdisplayskip}{1pt}
    \caption{Framework of parallel transaction execution based on asynchronous storage \label{algorithm:framework}}
    \KwData{transaction set $\mathcal{T}$,  memory node cache $\mathcal{C}_{node}$, memory state cache $\mathcal{C}_{state}$, no. of worker threads $\eta$, no. of retrieve threads $\zeta_r$, no. of commit threads $\zeta_c$, }
    \KwIn{State database $S$ }
    \KwOut{New state database $S$}

    $\texttt{Init}(\texttt{AsyStateRetrieval}, \zeta_{r})$ \texttt{/*\cref{algorithm:asyn_node_retrieval}*/}\\
    $\texttt{Init}(\texttt{AsyCommitStorage}, \zeta_c)$ \\
    $\texttt{Init}(\texttt{HashThread}, 1)$ \texttt{/*\cref{algorithm:pipeline_workflow}*/}\\
    $\texttt{Init}(\texttt{StoreThread}, 1)$  \texttt{/*\cref{algorithm:pipeline_workflow}*/}\\
    $I_{next}\gets 0$, $\mathcal{R}\gets \emptyset$, $\mathcal{W}\gets \emptyset$\\
    $S^*\gets S$\\
    \While{$len(\mathcal{T})\neq 0$}{
        \tcp{\cref{algorithm:batch_fetch}}
        $\mathcal{T}_{batch}\gets 
        \texttt{BatchFetching}(\mathcal{T},\eta)$ \\
        \tcp{\cref{algorithm:batch_parallel_exec}}
        $S^*, W, \mathcal{R}, \mathcal{W}, I_{next}\gets \texttt{BatchParallelExec}(\mathcal{T}_{batch}, S^*, \mathcal{T}, \mathcal{R}, \mathcal{W}, I_{next})$\\
        $\overline{W}\gets \texttt{update\_Merge}(\overline{W}, W)$\\
        $Accs\gets$ the set of distinguished accounts in $\overline{W}$\\
        \ForEach{ $A\in Accs$}{
            \If{\rm $A$ will not be accessed by transactions in $\mathcal{T}$ explicitly}{
                $U_A\gets$ updates to account $A$ in $\mathcal{W}$ \\
                $O\gets \langle A, U_A \rangle$\\
                $Q_{commit}.\texttt{Push}(O)$\\
            }
        }
        $\texttt{AsyCommitAccount}()$ \texttt{/*\cref{algorithm:pipeline_workflow}*/}\\
    }
    $\texttt{wait\_Storage\_Commit}()$\\
    $\texttt{AsyCommitAccount}()$ \texttt{/*\cref{algorithm:pipeline_workflow}*/}\\
    $\texttt{wait\_Account\_Commit}()$\\
    $S\gets S^*$\\
    \Return{$S$}\\
    \vspace{5pt}
    \Procedure{\rm $\texttt{AysCommitStorage}()$}{
        \While{\rm true}{
            $A,U_A\gets Q_{commit}.\texttt{Pop}()$\\
            $Trie_{A}\gets$ the MPT (contract storage) for account $A$\\ 
            $\texttt{write\_Updates}(Trie_A, U_A)$\\
            $root_A\gets Trie_A.\texttt{Commit}()$\\
            $Trie\gets$ the account trie\\
            $\texttt{update\_Account}(Trie, A)$\\
            
        }
    }
\end{algorithm}
\section{Analysis}\label{sec:analysis}

\subsection{Correctness}

To investigate the correctness of the \system's protocol, we show that it meets the deterministic
serializability criteria, given in \cref{theorem:correctness}. Intuitively, our protocol ensures that if a stale value of a state item is read, the execution of that transaction
will eventually be aborted and the incorrect results will be reverted. We first prove the following lemma and then prove the correctness of \system in \cref{theorem:correctness}.

\begin{lemma}\label{lemma:statedb}
    The state database $S$ always reflects the effects of executing transactions $\langle T_0, \ldots, T_{I_{\text{next}} - 1} \rangle$ serially after the parallel execution of each batch of transactions.
\end{lemma}

\begin{proof}
    This follows directly from the protocol design. Transactions are committed strictly in the order specified by the block. After each batch execution, the the state updates of each transaction $T_{p_i}$ are applied to the database $S$ sequentially, if $T_{p_i}$ does not conflict with proceeding transactions. This ensures that $S$ consistently reflects the cumulative effects of transactions $\langle T_0, \ldots, T_{I_{\text{next}} - 1} \rangle$.
\end{proof}

\begin{lemma}\label{lemma:correctness}
For each batch $\mathcal{T}_{\text{batch}} = \langle T_{p_i}, \ldots, T_{p_{\eta}} \rangle$, if the execution $\mathcal{E}_{p_i}$ of transaction $T_{p_i}$ reads a state item whose value is stale or not yet finalized in the state database, then $\mathcal{E}_{p_i}$ will eventually be aborted.
\end{lemma}

\begin{proof}
If the execution $\mathcal{E}_{p_i}$ reads a stale value, it must have missed an update present in the write set $\mathcal{W}$ produced by a preceding transaction that has not yet been committed. According to the abort policy of batch-parallel execution (see \cref{subsec:parallel_exec}), such an execution $\mathcal{E}_{p_i}$ must be aborted to maintain consistency and will be re-executed in a subsequent batch. Hence, the lemma holds for every transaction in the batch.
\end{proof}

\begin{lemma}\label{lemma:next}
    For the parallel execution of each batch $\mathcal{T}_{\text{batch}} = \langle T_{p_i}, \ldots, T_{p_{\eta}} \rangle$, before merging the states produced by all threads, the value of $I_{\text{next}}$ must be strictly equal to the smallest transaction index in the set $\mathcal{R}$ or $\mathcal{W}$.
\end{lemma}

\begin{proof}
    The index $I_{\text{next}}$ indicates the position of the next transaction to be committed. To ensure the correctness of state merging, all transactions with indices less than $I_{\text{next}}$ must have their read and write operations fully resolved and their effects finalized. If there exists a read or write operation with an index smaller than $I_{\text{next}}$, merging states prematurely could violate the consistency of the execution order. Therefore, $I_{\text{next}}$ must be set to the smallest index present in either $\mathcal{R}$ or $\mathcal{W}$ before state merging can safely proceed.
\end{proof}

\begin{lemma}\label{lemma:one_tx}
    For the execution of each batch $\mathcal{T}_{\text{batch}} = \langle T_{p_i}, \ldots, T_{p_{\eta}} \rangle$, at least one additional transaction is guaranteed to be successfully committed.
\end{lemma}

\begin{proof}
    During the parallel execution of the batch, before merging states, the value of $I_{\text{next}}$ must be equal to the smallest transaction index in the global read set $\mathcal{R}$, as established in \cref{lemma:next}. By construction, the transaction $T_{I_{\text{next}}}$ does not conflict with any other uncommitted transactions and can therefore be safely committed. This is guaranteed by the logic in Lines 13–17 of \cref{algorithm:batch_parallel_exec}, which ensures that such a transaction is committed first. Thus, at least one transaction in the batch will always be successfully committed.
\end{proof}

\begin{theorem}\label{theorem:correctness}
Given a block $B_l$ containing a sequence of transactions $\mathcal{T} = \langle T_1, \ldots, T_m \rangle$, the schedule produced by the parallel execution of \system yields the same final state as that produced by a serial execution of the transactions in order.
\end{theorem}

\begin{proof}

Let $\mathcal{T} = \langle T_1, \ldots, T_m \rangle$ be the sequence of transactions in block $B_l$, and let the parallel execution proceed in batches, with transactions being committed incrementally.
From \cref{lemma:correctness}, any transaction that reads a stale or unmerged state will be aborted and deferred for re-execution in a subsequent batch. Thus, only transactions that observe a consistent and up-to-date view of the state are allowed to commit.
According to \cref{lemma:one_tx}, each batch execution guarantees that at least one transaction is successfully committed. Therefore, all transactions in $\mathcal{T}$ will eventually be committed within a finite number of batches.
Finally, once all transactions have been committed, \cref{lemma:statedb} ensures that the final state stored in the state database is equivalent to the result of executing all transactions serially in the order $\langle T_1, \ldots, T_m \rangle$.
\end{proof}

\subsection{Commit Point Determination}\label{subsec:commit_point}
In this subsection, we discuss the method for determining the commit point for each node in the account trie. 

Let $n_r$ denote the number of remaining transactions to be executed. On average, each transaction updates $\mu$ accounts. Given that the account trie contains a substantial number of accounts (e.g., over one million), we assume that the first four levels (0 to 4) of the trie are fully populated. This assumption simplifies the analysis of commit points.
The commit point of a node $N$, denoted as $L$, is the point at which no further transaction in the block will modify $N$. At this stage, committing $N$ to the hashing phase is safe. However, predicting $L$ precisely is infeasible, so we estimate it as $L^*$. 

For a node $N$ at level $r$ ($\leq 4$), let $X^r$ denote the number of modifications to $N$ by the remaining $m$ transactions, amounting to an expected $u = \mu m$ state updates. The probability that a given state update modifies $N$ is $16^{-r}$, and thus, the probability that $N$ remains unmodified is:
\begin{align*}
    \small
    Pr[X^r=0|u] &=\left(1-16^{-r}\right)^{u}\\
    &=\left(1-16^{-r}\right)^{-16^r u (-16^{-r})} \approx e^{- \frac{u}{16^r}}  
\end{align*}
This probability increases as fewer updates remain. When $Pr[X^r=0|u]$ exceeds a predefined threshold $\alpha$ (e.g., 0.9), the likelihood of $N$ being modified again is sufficiently low, making it advantageous to hash $N$ asynchronously.

To determine a practical commit point, we analyze a specific case. When $r=4$ and $u \leq 4000$, the probability of no further modifications, $Pr[X^4=0|4000]$, is approximately 93.9\%, which is sufficiently high. Since the probability increases for $r \geq 4$ due to a larger number of nodes at deeper levels, we establish the following bound:
\begin{equation*}\label{equ:commit_p}
    \small
    Pr[X^r=0|u] \geq Pr[X^4=0|u] \geq Pr[X^4=0|4000] \approx 93.3\%
\end{equation*}
For nodes at level $r \geq 4$, we set the commit point $L^*$ to $4000$, the total number of transactions in the block. When $r \leq 3$, $L^*$ is estimated as follows:

\begin{equation}\label{equ:commit_point}
\small
\begin{aligned}
    Pr[X^r=0|u]  \geq \alpha  &\Rightarrow  e^{-\frac{u}{16^r}}\leq \alpha  \Rightarrow u \leq - 16^r \ln{\alpha} \\
    &\Rightarrow  m \leq \boxed{- \left( 16^r \ln{\alpha}\right) /\mu=L^*}
\end{aligned}
\end{equation}

According to \cref{equ:commit_point}, the commit points for levels $r \leq 1$ are close to zero, so we directly set them to 0. The final commit point determination is summarized as:
\begin{equation*}
\small
L^* =
\begin{cases}
    0 & \text{if } r\leq 1\\
   - \left(16^r \ln{\alpha}\right)/\mu & \text{if } 2\leq r\leq 3 \\
   4000 & \text{if } r\geq 4
\end{cases}
\end{equation*}

If node $N$ is a leaf node in the account trie, it stores the account state of an account $A$. If $A$ is still explicitly accessed by remaining transactions, $N$ cannot be hashed, even if it has reached its commit point. This is because $N$ must be modified again when $A$ is updated.

\subsection{Recovery}\label{subsec:recovery}
To support direct state reading, \system records state values in the database $D_{direct}$ in the form $\langle A||\kappa, v \rangle$ during the store phase of each state update. However, if a machine node crashes, only a subset of the state updates may have been inserted into $D_{direct}$, leading to potential inconsistencies. When the machine node recovers, it must replay the blocks to restore the correct state. In such cases, execution may read inconsistent state values from $D_{direct}$, compromising correctness.
To address this issue, we extend the record format to $\langle A||\kappa, v, l \rangle$, where $l$ represents the block height at which the state update occurred. When replaying from a specific block $B_l$, all state values with height greater than $l$ are discarded. This ensures that execution retrieves the correct state values from the account trie and storage trie, maintaining consistency across state updates.
By incorporating block height information, this approach guarantees that state values remain consistent after recovery, preventing partial writes from causing execution errors. Additionally, it enables efficient state restoration without requiring a full resynchronization of the database.
\section{Conclusion}  
\label{sec:conclusion}  

In this paper, we proposed \system, a batch-based parallel transaction execution framework that optimizes both computation and storage efficiency in blockchain systems. By integrating direct state reading, asynchronous parallel state loading, and a pipelined workflow, \system enables efficient parallel execution while mitigating the storage bottlenecks that limit existing approaches. These optimizations ensure scalable transaction processing while maintaining deterministic serializability. In future work, we plan to explore adaptive transaction scheduling and enhanced state caching mechanisms to further improve blockchain scalability.
{
\small
\bibliographystyle{plain}
\bibliography{ref}

\begin{thebibliography}{10}

\bibitem{financial}
{Blockchain Use Cases in Financial Services}.
\newblock \url{http: //blog.deloitte.com.ng/5-
  blockchain-use-cases-in-financial- services/}, 2017.

\bibitem{healthcare}
{Blockchain: Opportunities for Health Care}.
\newblock
  \url{https://www2.deloitte.com/us/en/pages/public-sector/articles/blockchain-opportunities-for-health-care.html},
  2018.

\bibitem{bcos}
{FISCO-BCOS}.
\newblock \url{http://fisco-bcos.org/}, 2020.

\bibitem{supplychain}
{IBM Blockchain for Supply Chain}.
\newblock \url{https://www.ibm.com/blockchain/supply-chain/}, 2020.

\bibitem{kecccak256}
{Keccak256}.
\newblock \url{https://godoc.org/golang.org/x/crypto/sha3}, 2021.

\bibitem{hyperledger}
Hyperledger.
\newblock \url{https://www.hyperledger.org}, 2024.

\bibitem{amiri2019parblockchain}
Mohammad~Javad Amiri, Divyakant Agrawal, and Amr El~Abbadi.
\newblock Parblockchain: Leveraging transaction parallelism in permissioned
  blockchain systems.
\newblock In {\em 2019 IEEE 39th International Conference on Distributed
  Computing Systems (ICDCS)}, pages 1337--1347. IEEE, 2019.

\bibitem{androulaki2018hyperledger}
Elli Androulaki et~al.
\newblock Hyperledger fabric: a distributed operating system for permissioned
  blockchains.
\newblock In {\em Proceedings of the thirteenth EuroSys conference}, pages
  1--15, 2018.

\bibitem{anjana2019efficient}
Parwat~Singh Anjana, Sweta Kumari, Sathya Peri, Sachin Rathor, and Archit
  Somani.
\newblock An efficient framework for optimistic concurrent execution of smart
  contracts.
\newblock In {\em 2019 27th Euromicro International Conference on Parallel,
  Distributed and Network-Based Processing (PDP)}, pages 83--92. IEEE, 2019.

\bibitem{eswaran1976notions}
Kapali~P. Eswaran, Jim~N Gray, Raymond~A. Lorie, and Irving~L. Traiger.
\newblock The notions of consistency and predicate locks in a database system.
\newblock {\em Communications of the ACM}, 19(11):624--633, 1976.

\bibitem{feist2019slither}
Josselin Feist, Gustavo Grieco, and Alex Groce.
\newblock Slither: a static analysis framework for smart contracts.
\newblock In {\em 2019 IEEE/ACM 2nd International Workshop on Emerging Trends
  in Software Engineering for Blockchain (WETSEB)}, pages 8--15. IEEE, 2019.

\bibitem{garamvolgyi2022utilizing}
P{\'e}ter Garamv{\"o}lgyi, Yuxi Liu, Dong Zhou, Fan Long, and Ming Wu.
\newblock Utilizing parallelism in smart contracts on decentralized blockchains
  by taming application-inherent conflicts.
\newblock {\em arXiv preprint arXiv:2201.03749}, 2022.

\bibitem{kung1981optimistic}
Hsiang-Tsung Kung and John~T Robinson.
\newblock On optimistic methods for concurrency control.
\newblock {\em ACM Transactions on Database Systems (TODS)}, 6(2):213--226,
  1981.

\bibitem{kwon2014tendermint}
Jae Kwon.
\newblock Tendermint: Consensus without mining.
\newblock {\em Draft v. 0.6, fall}, 1(11), 2014.

\bibitem{li2020decentralized}
Chenxin Li, Peilun Li, Dong Zhou, Zhe Yang, Ming Wu, Guang Yang, Wei Xu, Fan
  Long, and Andrew Chi-Chih Yao.
\newblock A decentralized blockchain with high throughput and fast
  confirmation.
\newblock In {\em 2020 USENIX Annual Technical Conference (USENIX ATC 20)},
  pages 515--528, 2020.

\bibitem{li2023lvmt}
Chenxing Li, Sidi~Mohamed Beillahi, Guang Yang, Ming Wu, Wei Xu, and Fan Long.
\newblock {LVMT}: An efficient authenticated storage for blockchain.
\newblock In {\em 17th USENIX Symposium on Operating Systems Design and
  Implementation (OSDI 23)}, pages 135--153, 2023.

\bibitem{merkle1989certified}
Ralph~C Merkle.
\newblock A certified digital signature.
\newblock In {\em Conference on the Theory and Application of Cryptology},
  pages 218--238. Springer, 1989.

\bibitem{minglani2017kinetic}
Manas Minglani, Jim Diehl, Xiang Cao, Binghze Li, Dongchul Park, David~J Lilja,
  and David~HC Du.
\newblock Kinetic action: Performance analysis of integrated key-value storage
  devices vs. leveldb servers.
\newblock In {\em 2017 IEEE 23rd International Conference on Parallel and
  Distributed Systems (ICPADS)}, pages 501--510. IEEE, 2017.

\bibitem{nakamoto2008bitcoin}
Satoshi Nakamoto et~al.
\newblock Bitcoin: A peer-to-peer electronic cash system.
\newblock 2008.

\bibitem{2019Blockchain}
Senthil Nathan, Chander Govindarajan, Adarsh Saraf, Manish Sethi, and Praveen
  Jayachandran.
\newblock Blockchain meets database: Design and implementation of a blockchain
  relational database.
\newblock {\em Proceedings of the VLDB Endowment}, 12(11):1539--1552, 2019.

\bibitem{o1996log}
Patrick O’Neil, Edward Cheng, Dieter Gawlick, and Elizabeth O’Neil.
\newblock The log-structured merge-tree (lsm-tree).
\newblock {\em Acta Informatica}, 33(4):351--385, 1996.

\bibitem{sharma2019blurring}
Ankur Sharma, Felix~Martin Schuhknecht, Divya Agrawal, and Jens Dittrich.
\newblock Blurring the lines between blockchains and database systems: the case
  of hyperledger fabric.
\newblock In {\em Proceedings of the 2019 International Conference on
  Management of Data}, pages 105--122, 2019.

\bibitem{solidity}
{Solidity}.
\newblock \url{https://docs.soliditylang.org/}, 2022.

\bibitem{wood2014ethereum}
Gavin Wood et~al.
\newblock Ethereum: A secure decentralised generalised transaction ledger.
\newblock {\em Ethereum project yellow paper}, 151(2014):1--32, 2014.

\bibitem{yu2020ohie}
Haifeng Yu, Ivica Nikoli{\'c}, Ruomu Hou, and Prateek Saxena.
\newblock Ohie: Blockchain scaling made simple.
\newblock In {\em 2020 IEEE Symposium on Security and Privacy (SP)}, pages
  90--105. IEEE, 2020.

\bibitem{zhang2018enabling}
An~Zhang and Kunlong Zhang.
\newblock Enabling concurrency on smart contracts using multiversion ordering.
\newblock In {\em Asia-Pacific Web (APWeb) and Web-Age Information Management
  (WAIM) Joint International Conference on Web and Big Data}, pages 425--439.
  Springer, 2018.

\end{thebibliography}
}

\end{document}